\documentclass[11pt,a4paper]{article}
\usepackage[T1]{fontenc}
\usepackage[utf8]{inputenc}
\usepackage{fullpage}
\usepackage{lmodern}
\usepackage{amsmath,amssymb,amsthm}
\usepackage{microtype}
\usepackage[hidelinks]{hyperref}
\usepackage{bookmark}
\usepackage{enumitem}
\usepackage{xcolor}
\usepackage{graphicx}
\newtheorem{lemma}{Lemma}[section]
\newtheorem{theorem}[lemma]{Theorem}
\newtheorem{corollary}[lemma]{Corollary}
\numberwithin{equation}{section}

\begin{document}

\title{Logarithmic Separation between LZ77 and LZ-End and Grammar Compression}
\author{Philip Bille \and Inge Li Gørtz \and Rikke Jessen \and Jingxun Li \and Simon Rumle Tarnow \and Renfei Zhou}
\date{}
\maketitle

\section*{Introduction}
The Lempel-Ziv factorization of a text~\cite{lz77} (LZ77) is a greedy left-to-right parse in maximal substrings such that each substring is either a single (new) character or already occurred to the left. Despite its simplicity, LZ77 can be easily shown to be optimal among all unidirectional parses (i.e. that copy phrases from left-to-right), and dominates other popular compression schemes such as Straight Line Programs (SLPs), i.e., context-free grammars that generate only the text as output.

LZ-End~\cite{KN10,KN13} is a variant of LZ77 parsing with the added constraint that the source of the phrase must end at a phrase boundary of a previous phrase. Let $z(S)$ be the number of phrases in the LZ77 parsing of a string $S$, and let  $z_\mathrm{no}(S)$ be the number of phrases when overlaps are not allowed between phrases and their sources. Lastly, let $z_{\mathrm{end}}(S)$ be the minimum number of phrases in any legal LZ-End parsing of $S$.
While there is no known data structure of $O(z(S))$ size that provides eﬃcient random access to arbitrary positions in the string, it was recently shown that  $\tilde{O}(1)$ time access could be achieved with $O(z_{\mathrm{end}}(S))$ space~\cite{BonehGawrychowski2026}.
Note that, every LZ-End parsing is also a legal unidirectional parse, so
$z_{\mathrm{no}}(S)\le z_{\mathrm{end}}(S)$.

In this paper we answer the following question: 
How much can the source-boundary restriction increase the number of phrases?

\paragraph{Our results.}
We give an explicit, deterministic family of binary strings $T_k$, indexed
by $k$, where $k$ is a power of two, such that
\[
 |T_k|=\Theta(k^2\log k),\qquad
 z_{\mathrm{no}}(T_k)=O(k),\qquad
 z_{\mathrm{end}}(T_k)=\Omega(k\log k).
\]
Writing $n=|T_k|$, we obtain
\[
 \frac{z_{\mathrm{end}}(T_k)}{z_{\mathrm{no}}(T_k)}
 =\Omega(\log n).
\]
The lower bound holds for every legal LZ-End parsing, with unrestricted choices of source occurrences and phrase boundaries. 

The result has a consequence for grammar compression. Let $g(S)$ denote the
minimum total length of the right-hand sides of an acyclic grammar generating
only $S$. The grammar-to-LZ77 conversion of Charikar et al.~\cite{Charikar2005}
also produces a legal LZ-End parsing: expand a variable at its first
occurrence and copy each subsequent occurrence from that first expansion,
which ends at a phrase boundary. Thus $z_{\mathrm{end}}(S)\le g(S)$, and our
strings also satisfy
\[
 \frac{g(T_k)}{z_{\mathrm{no}}(T_k)}=\Omega(\log n).
\]
The same asymptotic separation holds when grammar size is measured by the
number of variables in a binary straight-line program (SLP). We give the
conversion argument explicitly in Section~\ref{sec:grammar-consequence}.
Together with the standard LZ77-to-grammar upper bound
$g(S)=O(z_{\mathrm{no}}(S)\log n)$, these inequalities show that the
worst-case ratios to non-overlapping LZ77 are $\Theta(\log n)$ for both
optimal LZ-End and smallest grammars.

\paragraph{The Snapshot String Construction.}
We first work over a growing alphabet. Start with a string $A_0$ of $k$ distinct symbols. Change one position at a time to a fresh symbol to get $A_1$, choosing the changing positions in \emph{bit reversal order}~\cite{PD2006,Wilber1989,HF1998,FS2026}, and concatenate $A_0 \ldots A_{k}$ with fresh separators between them. Each version differs from its predecessor at exactly one position. LZ77 can therefore represent the string by a prefix copy, the new symbol, and a suffix copy, giving $O(k)$ phrases in total.

For the LZ-End lower bound, place a complete balanced binary tree over the string positions in $A_0$. At each node, consider consecutive updates within its interval where the first update is in  the left child and the next is in the right child. Bit reversal produces $k/2$ such events per tree level, for a total of $(k\log_2 k)/2$. We show that each event requires a distinct phrase in every LZ-End parsing.

The charging argument finds a phrase between the two updates of each event.
Between two such updates of symbols in the interval of a node $v$, choose the smallest position $j$ in the right child of $v$ that is a phrase boundary in some $A_i$ between the two updates, and then the smallest $i$ such that $A_i$ contains that
boundary. The preceding boundary lies in the left child. If its phrase started before the position changed by the left update, the phrase would
contain that update's fresh symbol. A source for the phrase would then have to occur after that update and end at $j$. This contradicts
the choice of the smallest $i$ such that $A_i$ contains that boundary. 
Hence the two phrase boundaries lie in opposite children of the node, and between the two update times. The lowest common ancestor of the two boundary coordinates identifies the node uniquely, and since the row intervals of events at the same node are disjoint, no phrase can receive two charges.

To obtain the result for binary strings, we encode the symbols with fixed-length codewords whose headers identify codeword starts. A fresh codeword forces a phrase boundary strictly inside it. The charging argument can therefore be repeated in bit coordinates, without assuming that parsing boundaries align with codewords. Any copied phrase containing a complete codeword for an updated symbol identifies both its array position and the alignment of its source.
Finally, numbering fresh symbols in order of appearance allows each new codeword to be represented by a constant number of non-overlapping LZ77 copies. The encoding preserves the $O(k)$ LZ77 bound and the $\Omega(k\log k)$ LZ-End lower bound.

Finally, we show that the result also holds for the version of LZ-End and LZ77 where a phrase consists of a copy and an extra character. This follows immediately from the fact that the two versions of the parsing only differ by a factor of 2 (for completeness, we prove this in Section~\ref{sec:appended-character}.

\subsection*{Related work}
Our snapshot string construction is closely related to the bit-reversal block family of strings of Fujie and Shibata~\cite{FS2026}. Their blocks distinguish three states (before, at, and after its transition) at each coordinate. 
Identifying the at-threshold state with the before-transition state produces exactly our intermediate snapshots, and their  ordered block string becomes the central portion of our construction. 
They use this family to establish sensitivity bounds and separations for several restricted compression schemes. They do this by bounding the length of a copy phrase in terms of the distance between its block and the source block. 
Our lower bound 
exploits the placement of LZ-End phrase boundaries: each left-to-right transition within a dyadic coordinate interval forces a distinct phrase, independently of how the parsing is chosen. We additionally establish that this charging argument survives binary encoding.

Kreft and Navarro~\cite{KN10,KN13} introduced LZ-End to combine LZ77-like compression with
fast extraction, and subsequently developed a self-index based on the
parsing.
Their original convention appends one character to a copy whose source ends
at a phrase boundary. 
Kreft and Navarro gave a growing-alphabet family for which the ratio between
LZ-End and non-overlapping LZ77 approaches two. Ideue et al.~\cite{Ideue2021} obtained the
same limiting ratio over a binary alphabet using the period-doubling
sequence.
These comparisons concern greedy parsings with an appended character;
their exact phrase counts should therefore be distinguished from ours.

Kempa and Saha~\cite{KempaSaha2022} proved a greedy LZ-End upper bound of $O(z\log^2 n)$,
where $z$ is the LZ77 phrase count with overlapping sources, and gave
linear-space data structures for random access.
Gawrychowski, Kosche, and Manea improved the phrase-count bound to
$O(z\log^2 n/\log\log n)$~\cite{GawrychowskiKoscheManea2023}.
These results concern greedy LZ-End. The logarithmic upper bound relevant
to our result concerns the optimal parsing and follows through grammar
compression; it makes no assertion about the performance of the greedy
LZ-End algorithm.

Bannai et al. studied optimal LZ-End under the whole-copy convention. They
proved that finding an optimal parsing is NP-hard, gave a MAX-SAT
formulation, and constructed binary strings whose greedy-to-optimal phrase
ratio approaches two~\cite{Bannai2023}.
Their separation measures the cost of greedy phrase selection. Our result
measures the source-boundary restriction itself: even the smallest legal
LZ-End parsing is a logarithmic factor larger than non-overlapping LZ77.

For grammar compression, the basic comparison with non-overlapping LZ77
comes from the conversions of Rytter~\cite{Rytter2003} and Charikar et al.~\cite[Sections~VII-D--VII-G]{Charikar2005}. A grammar of
size $g$ yields a non-overlapping parsing with at most $g$ phrases,
whereas a parsing with $z_{\mathrm{no}}$ phrases can be converted into a
grammar of size
$O(z_{\mathrm{no}}(1+\log(n/z_{\mathrm{no}})))$.
Je\.{z} later gave a simpler linear-time conversion with the same
asymptotic size guarantee~\cite{Jez2016}.
These results give the logarithmic upper bound on $g/z_{\mathrm{no}}$
used here.

In the other direction, Charikar et al. used a connection with addition
chains to construct strings satisfying
$g/z_{\mathrm{no}}=\Omega(\log n/\log\log n)$.
Hucke, Lohrey, and Reh obtained the same asymptotic separation over a
binary alphabet~\cite{HuckeLohreyReh2016}.
Bille, Gagie, G\o rtz, and Prezza extended the
$\Omega(\log n/\log\log n)$ separation to run-length SLPs, which also
allow rules of the form $X\to Y^k$~\cite{Bille2018}.
Kosolobov and Valenzuela explicitly identified the remaining
$\log\log n$ gap between the known grammar--LZ77 bounds as an open
problem~\cite[Sections~1 and~5]{KosolobovValenzuela2021}.
For ordinary grammars, Corollary~\ref{cor:grammar-separation} improves
the lower-bound factor to $\Omega(\log n)$ over a binary alphabet,
matching the logarithmic upper bound. To our knowledge, this tight
separation from non-overlapping LZ77 was not previously known.

The non-overlap restriction is essential to this comparison. For the
unary string $a^n$, overlapping LZ77 already uses $O(1)$ phrases while
the smallest grammar has size $\Theta(\log n)$; non-overlapping LZ77
uses $\Theta(\log n)$ phrases on the same string. Thus a logarithmic
grammar separation from overlapping LZ77 alone does not settle the
comparison studied here.

Our construction uses successive single-character substitutions, but the
comparison differs from sensitivity results. Akagi, Funakoshi, and Inenaga
proved tight constant bounds on the multiplicative sensitivity of several
LZ77 variants under a single edit~\cite{Akagi2023}, while Fujie et al. established tight
$\Theta(n^{2/3})$ worst-case additive sensitivity bounds for overlapping and
non-overlapping LZSS and for optimal LZ-End~\cite{Fujie2025}.
Here we concatenate all snapshots and compare different compression
measures on the resulting single string.

\paragraph{AI Declaration} This result was obtained with the help of AI, more precisely OpenAI ChatGPT Astra and Sol 6. The result was obtained through a longer discussion starting with another window-copying construction from the authors that together with the ruler string as basis, and a later a prefix window-copying construction combined with the ruler string, that ChatGPT Sol 6 proved gave a $\Omega(\log n / \log\log n)$ separation. The final string and proof was found by ChatGPT Astra. All proofs and constructions have been verified by the authors. 

\paragraph{Organization.}
Section~\ref{definitions-non-overlapping-lz77-and-lz-end} fixes the parsing
conventions. Section~\ref{sec:results} states the separation and its grammar
consequence. Section~\ref{sec:snapshot-construction} gives the
growing-alphabet construction and the charging argument.
Section~\ref{sec:binary-charging} establishes the binary LZ77 upper bound
and LZ-End lower bound. Section~\ref{sec:main-proof} combines these bounds
to prove the main theorem. Section~\ref{sec:appended-character} contains the extension to the version where copy phrases contains an extra character.

\section{Definitions}
\label{definitions-non-overlapping-lz77-and-lz-end}

All logarithms are base two. 
\paragraph{Strings, phrases, and boundaries}
\label{strings-intervals-and-boundaries}

A string~$S$ of length~$|S|=n$ is a sequence $S[0]\cdots S[n-1]$
of characters drawn from an alphabet $\Sigma$. The string
$S[i]\cdots S[j]$, denoted $S[i:j]$, is a substring of $S$.
By convention, $S[a:b)=S[a:b-1]$ for $a<b$, and $S[a:a)$ is empty.
A parsing partitions $S$ into nonempty phrases $f_1\cdots f_z$, where
\[
0=p_1<p_2<\cdots<p_{z+1}=n,\qquad f_i=S[p_i:p_{i+1}).
\]
A phrase boundary is a prefix length $p_i$; in particular, $p_1=0$
is the initial boundary. The end of phrase $f_i$ is $p_{i+1}$.

\paragraph{LZ77.}
The non-overlapping LZ77 parsing is formed greedily from left to right.
The $i$th phrase is either (a) the first occurrence of a character in
$S$, in which case it has length one, or (b) the longest nonempty
substring starting at $p_i$ with an occurrence ending to the left of
$p_i$. Precisely, a copied length $\ell$ is admissible when
\begin{equation}\label{eq:lz77-source}
1\le\ell\le n-p_i,\qquad
\exists q\ge0:\quad q+\ell\le p_i,
\quad S[q:q+\ell)=S[p_i:p_i+\ell).
\end{equation}
Here ``ending to the left'' refers to the last character position:
the exclusive source endpoint $q+\ell$ may equal $p_i$.
A legal non-greedy non-overlapping LZ77 parsing obeys the same rules
but need not choose the longest admissible copy. It is well-known that the greedy parsing is optimal for LZ77.

\paragraph{LZ-End.}
An LZ-End parsing partitions $S$ into phrases $f_1\cdots f_z$ with
the same boundary notation. Each phrase $f_i$ is either (a) the first
occurrence of a character, of length one, or (b) a substring with an
earlier occurrence ending at a phrase boundary. Thus a copied phrase
has a source boundary $p_j$ satisfying
\begin{equation}\label{eq:lzend-source}
1\le j\le i,\qquad |f_i|\le p_j,\qquad
f_i=S[p_i:p_{i+1})=S[p_j-|f_i|:p_j).
\end{equation}
The inequality $|f_i|\le p_j$ excludes negative source positions.
The case $j=i$ is permitted: the source may end at the current phrase
start. It is nevertheless wholly before the target. Sources may cross
several phrases; only their right endpoint must be a boundary.
The final phrase obeys exactly the same rules as every other phrase.

The \emph{greedy LZ-End parsing} chooses the longest admissible copied
phrase whenever a copy is available. If the next character has not
occurred before, it forms a one-character phrase. These rules always
permit progress: a character's first occurrence cannot lie inside a
copied phrase, so it ends a one-character phrase; that endpoint is an
available source for any later one-character copy of the character.
Admissible copied lengths need not form an interval.

Greedy LZ-End is not always optimal. For example, on $S=00100010$ its
phrases are
\[
0\mid0\mid1\mid00\mid01\mid0,
\]
whereas $0\mid0\mid1\mid0\mid0010$ is a legal five-phrase parsing.
The last phrase in the latter copies $S[0:4)$, ending at the boundary
$p_5=4$. In the greedy parsing, at position $3$ the longest copy is
$00$ ending at boundary $2$; at position $5$ the longest copy is $01$
ending at boundary $3$. The first three phrases are forced in every
legal parsing, and the remaining five characters cannot be copied
non-overlapping from a prefix of length three, so five phrases are optimal. 

\paragraph{Number of phrases.}
Let $z_{\mathrm{no}}(S)$ be the non-overlapping LZ77 phrase count, and define
\[
z_{\mathrm{end}}(S)=\min\{z:S\text{ has a legal LZ-End parsing with }z\text{ phrases}\}.
\]
Let $z_{\mathrm{end}}^{\mathrm{greedy}}(S)$ denote the greedy LZ-End phrase count.
Then
\[
z_{\mathrm{no}}(S)\le z_{\mathrm{end}}(S)\le z_{\mathrm{end}}^{\mathrm{greedy}}(S).
\]
The first inequality follows from the optimality of the greedy LZ77 parsing
since every legal LZ-End parsing is a legal non-overlapping LZ77 parsing.
The second follows by taking a minimum over legal LZ-End parsings.

\section{Results}
\label{sec:results}

Our first result is a separation between LZ77 and LZ-End. 
\begin{theorem}\label{thm:snapshot-separation}
 For every power of two \(k=2^h\), there is a binary string
\(T_k\) of length \(c(k+1)^2\), where
\(c=2\lceil\log_2(3k+1)\rceil+6\), such that
\[
z_{\rm no}(T_k)\le c+11k,
\qquad
z_{\rm end}(T_k)\ge \frac{k\log_2 k}{2}.
\]
Consequently \(z_{\rm end}/z_{\rm no}=\Omega(\log n)\), even over a binary
alphabet.   
\end{theorem}

\subsection{Consequence for Grammar Compression}
\label{sec:grammar-consequence}

An acyclic grammar has one production for each variable, nonempty right-hand sides, and a designated start variable generating a single string. Its size is the total number of symbols on all right-hand sides. A binary SLP restricts productions to $X\to YZ$ or $X\to a$.

The following lemma uses the grammar-to-LZ77 conversion of Charikar
et al.~\cite{Charikar2005}.
We observe that the same conversion also satisfies the LZ-End requirement
that every copied source end at a phrase boundary.

\begin{lemma}[Grammar-to-LZ-End conversion]
\label{lem:grammar-to-lzend}
Let an acyclic grammar of size $g$ with $r$ reachable variables generate a
nonempty string $S$. Then $S$ has a legal whole-copy LZ-End parsing with
at most $g-r+1\le g$ phrases. If the grammar is a binary SLP with $r$
variables, at most $r$ phrases suffice.
\end{lemma}
\begin{proof}
Traverse the derivation tree from left to right. At the first occurrence
of each variable, expand its production recursively. At every subsequent
occurrence, output its complete expansion as one copied phrase, using the
first occurrence as its source. At terminal leaves, output a singleton,
as a literal if it is the character's first occurrence and otherwise as
a copy of its first occurrence.

The first expansion of a variable is finished before a later occurrence
is visited: otherwise that later occurrence would be a descendant of the
first, contradicting acyclicity. Its source is therefore non-overlapping.
Moreover, its first expansion is a concatenation of complete output
phrases, so its right endpoint is a phrase boundary. First occurrences
of terminal characters are singleton phrases as well. Thus every output
phrase is legal for LZ-End.

In the pruned derivation tree, each reachable variable is expanded exactly
once. There are $r$ internal nodes and at most $g$ edges, so there are at most
$g-r+1$ leaves, one for each output phrase. For a binary SLP, if $r_2$ variables
have two-symbol right-hand sides and $r_1$ have terminal right-hand sides,
this count is $r_2+1\le r_2+r_1=r$, since $S$ is nonempty and hence
$r_1\ge1$.
\end{proof}

\begin{corollary}
\label{cor:grammar-separation}
For the binary strings $T_k$ of Theorem~\ref{thm:snapshot-separation}, every
acyclic grammar generating $T_k$ has size at least $k\log_2 k/2$.
Every binary SLP generating $T_k$ also has at least $k\log_2 k/2$ variables.
In particular, for $n=|T_k|$,
\[
 \frac{g(T_k)}{z_{\mathrm{no}}(T_k)}=\Omega(\log n).
\]
\end{corollary}
\begin{proof}
Apply Lemma~\ref{lem:grammar-to-lzend} to the LZ-End lower bound in
Theorem~\ref{thm:snapshot-separation}, and use
$z_{\mathrm{no}}(T_k)\le c+11k$ and $\log n=\Theta(\log k)$.
\end{proof}

\section{Construction over a Growing Alphabet}
\label{sec:snapshot-construction}
We first show a construction with a large alphabet. Let \(k = 2^h\). Our alphabet has size $3k+1$ and consists of the following mutually distinct symbols: 

\[
a_1,\ldots,a_k,\;
b_1,\ldots,b_k,\;
\#_0,\ldots,\#_k.
\]

Let \(p_1,\ldots,p_k\) be a permutation of \([1,k]\). We will define the precise permutation later. 
We construct our string $\mathcal{S}$ as follows.
Initially
\[
A_0 = a_1\cdots a_k.
\]
At iteration \(t\), obtain \(A_t\) from \(A_{t-1}\) by replacing
\(a_{p_t}\) at position \(p_t-1\) with \(b_{p_t}\). Thus each position
is changed exactly once, and \(A_k=b_1b_2\cdots b_k\). Define

\[
\mathcal{S} = A_0 \#_0 A_1 \#_1 \cdots A_k \#_k\;.
\]

The length of $\mathcal{S}$ is \((k+1)^2\). We call the substring $A_i\#_i$ the $i$th row, for $0\le i\le k$; it starts at string position $i(k+1)$. See Figure~\ref{fig:string_illustration} for an illustration of the construction.
A symbol at position $j'$ in the string is in column $j$ if it is the $j$th symbol in the row, that is $j' \equiv j-1 \pmod{k+1}$. A boundary at column $j\in[0,k+1]$ in row $i$ is a cut after the first $j$ symbols of the row, at prefix length $i(k+1)+j$. Thus a phrase in row $i$ with \emph{boundary columns} $l,r$ contains the symbols in columns $l+1,\ldots,r$ of row $i$.

\begin{figure}
    \centering    \includegraphics[width=0.8\linewidth]{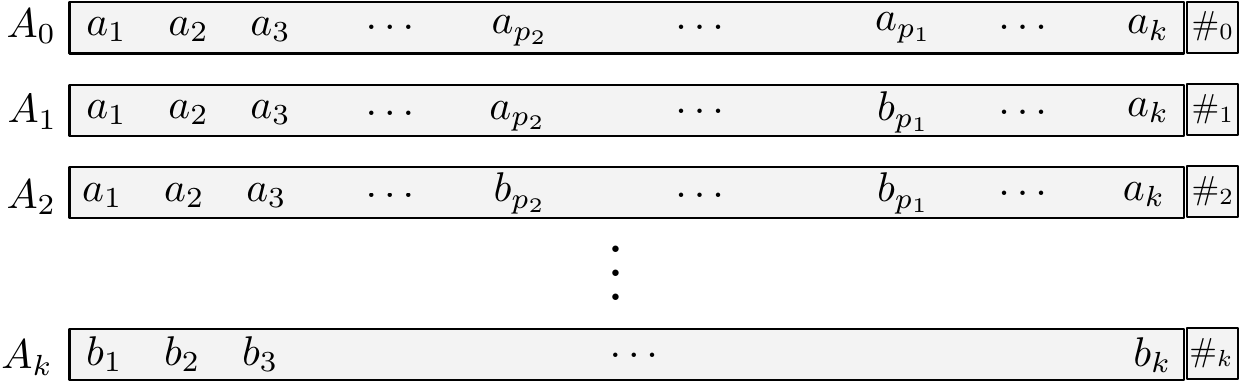}
    \caption{Illustration of the strings $A_0...A_k$. Each $A_i$ is thought of as a row, and the positions within the rows are thought of as columns.}
    \label{fig:string_illustration}
\end{figure}

\paragraph{LZ77 parsing of $\mathcal{S}$}
There is a legal non-overlapping LZ77 parsing
with at most \(5k+1\) phrases: the initial row and separator take \(k+1\)
literals.  Each later row can be parsed with at most 4 phrases: (i) its common prefix with the preceding row,
(ii) the new literal \(b_{p_t}\), (iii) its common suffix with the preceding row, and (iv) a phrase for the new separator. 
Thus,
\begin{equation}\label{eq:z_no-large-alphabet}
    z_{\mathrm{no}} \leq 5k+1\;.
\end{equation}

\paragraph{Properties of the string $\mathcal{S}$} Note that every separator is a forced singleton in every legal LZ-End parsing: it is new and has length one, and cannot occur inside a copied phrase. Consequently each row has its own partition into phrases. Moreover, in row \(t\),
the new \(b_{p_t}\) is a forced singleton, so columns \(p_t-1\) and \(p_t\) are boundaries of phrases in this row.
%
Every copied phrase inside a row must have a source inside a strictly earlier row and at exactly the same coordinates. This is true, since every symbol identifies its coordinate uniquely and a source cannot cross a separator. 

\subsection{LZ-End Lower Bound}
\label{sec:interleaving}

We first prove the following lemma. Let $\mathcal{T}_k$ be the complete balanced binary tree on coordinates \([1,k]\). Each node $v$ corresponds to an interval of coordinates $[l_v,r_v]$, i.e., the nodes correspond to all the dyadic intervals of \([1,k]\).

\begin{lemma}\label{lem:consecutive_updates}
   Let \(v=[i,j]\) be a node in $\mathcal{T}_k$, with children \([i,m]\) and \([m+1,j]\). Let \(s<t\) be consecutive times in the update subsequence restricted to \(v\), and suppose
    \[
    p_s \in [i,m],
    \qquad
    p_t \in [m+1,j].
    \] 
    Then every legal LZ-End parsing of \(\mathcal{S}\) contains a  phrase  in a row \(u \in [s,t]\) with boundary columns $(l, r)$ such that
    \begin{equation}\label{eq:growing-charged-phrase}
    p_s \le l \le m < r \le p_t.
    \end{equation}
    \end{lemma}




\begin{proof}
     Row \(t\) has the forced boundary \(p_t\) in the right child. Among all boundaries in the right child occurring in rows \(s,\ldots,t\), choose the smallest column \(r\) containing such a boundary. Then choose the first row \(u\) in \([s,t]\) containing a phrase boundary in column r. We have \(r\le p_t\). Its predecessor boundary \(l\) in row \(u\) satisfies \(l\le m\), by minimality of \(r\). See Figure~\ref{fig:interval_under_node}. 
     
     It remains to show that \(l \geq p_s\). We do this by contradiction. Suppose \(l<p_s\). Since \(r>m\ge p_s\), the phrase has length at least two and therefore is copied. It contains \(b_{p_s}\), since coordinate \(p_s\) was updated at \(s\) and is not updated again. Its source is in a row \(w<u\) at the same coordinates, hence its source endpoint is a boundary at column \(r\) in row \(w\). Since \(b_{p_s}\) is new at time \(s\), necessarily \(w\ge s\). This contradicts the choice of \(u\) as the first row in \([s,t]\) containing boundary \(r\). Thus \(l\ge p_s\), proving \eqref{eq:growing-charged-phrase}.
\end{proof}

\begin{figure}
    \centering
    \includegraphics[width=0.5\linewidth]{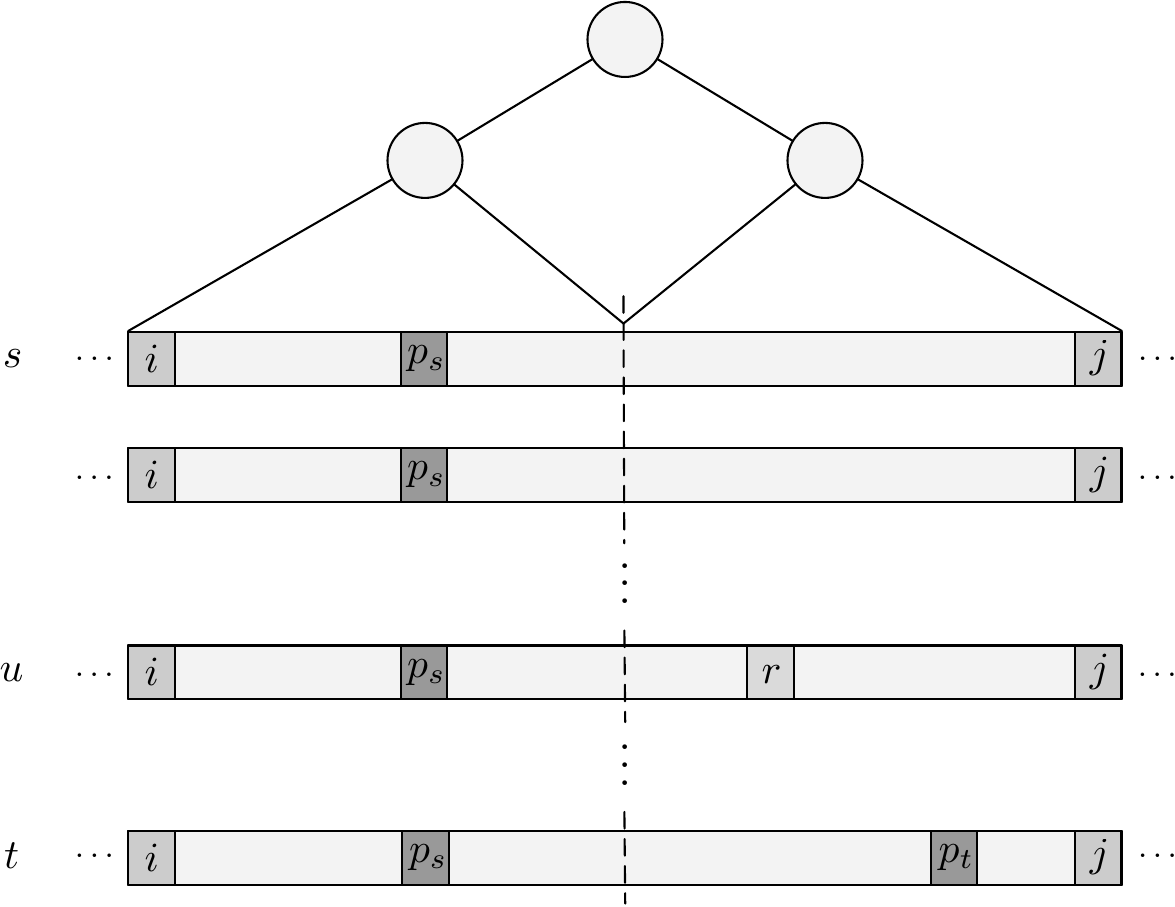}
    \caption{Illustration of the situation in Lemma~\ref{lem:consecutive_updates}. All updates to the string strictly between rows $s$ and row $t$ happen outside the interval of $v$.}
    \label{fig:interval_under_node}
\end{figure}

We call two consecutive updates $p_s$ and $p_t$ in the interval of a node $v$ an event $(v,s,t)$ if $p_s$ and $p_t$ belong to the intervals of the left and right child of $v$, respectively.

We will give a permutation that induces a large number of  such events \((v,s,t)\). To show a lower bound on the number of phrases we charge the event \((v,s,t)\) to the phrase found in the proof of the Lemma~\ref{lem:consecutive_updates}. For fixed \(v\), the intervals \([s,t]\) belonging to consecutive events are disjoint. Furthermore, a phrase satisfying \eqref{eq:growing-charged-phrase} has $l$ and $r$ in opposite children of exactly one node from $\mathcal{T}_k$, namely, the lowest common ancestor of coordinate keys \(l\) and \(r\) in $\mathcal{T}_k$. Thus no phrase is charged twice.

\paragraph{Permutation}
\label{sec:bit-reversal}
Recall that, $k=2^h$. We define the permutation as follows. 
Let 
\begin{align}\label{eq:permutation}
    p_t = 1 + \mathrm{reverse}_h(t-1),
\end{align}

where \(\mathrm{reverse}_h\) reverses \(h\) binary bits. See Figure~\ref{fig:string_as_matrix} for an illustration of the string construction using this permutation.
The permutation has the following property:
At every internal node $v$ in $\mathcal{T}_k$, the update subsequence restricted to $v$ alternates between
its children, starting in the left child.

\noindent Thus, a node of size \(2^q\) has
\(2^{q-1}\) left-to-right events. Summing over all nodes we get, 
\[
\sum_{d=0}^{h-1} 2^d \cdot \frac{k}{2^{d+1}} = \frac{kh}{2}\;.
\]

\noindent Therefore every legal LZ-End parsing of \(\mathcal{S}\) has at least

\[
\frac{kh}{2}
\]
phrases. Combined with \(z_{\rm no}(\mathcal{S})\le5k+1\) and \(n=(k+1)^2\), this gives an
\(\Omega(\log n)\) separation using a large alphabet. In the next section we show how to binary encode the alphabet symbols while retaining the lower bound.

\begin{figure}
    \centering
    \includegraphics[width=0.35\linewidth]{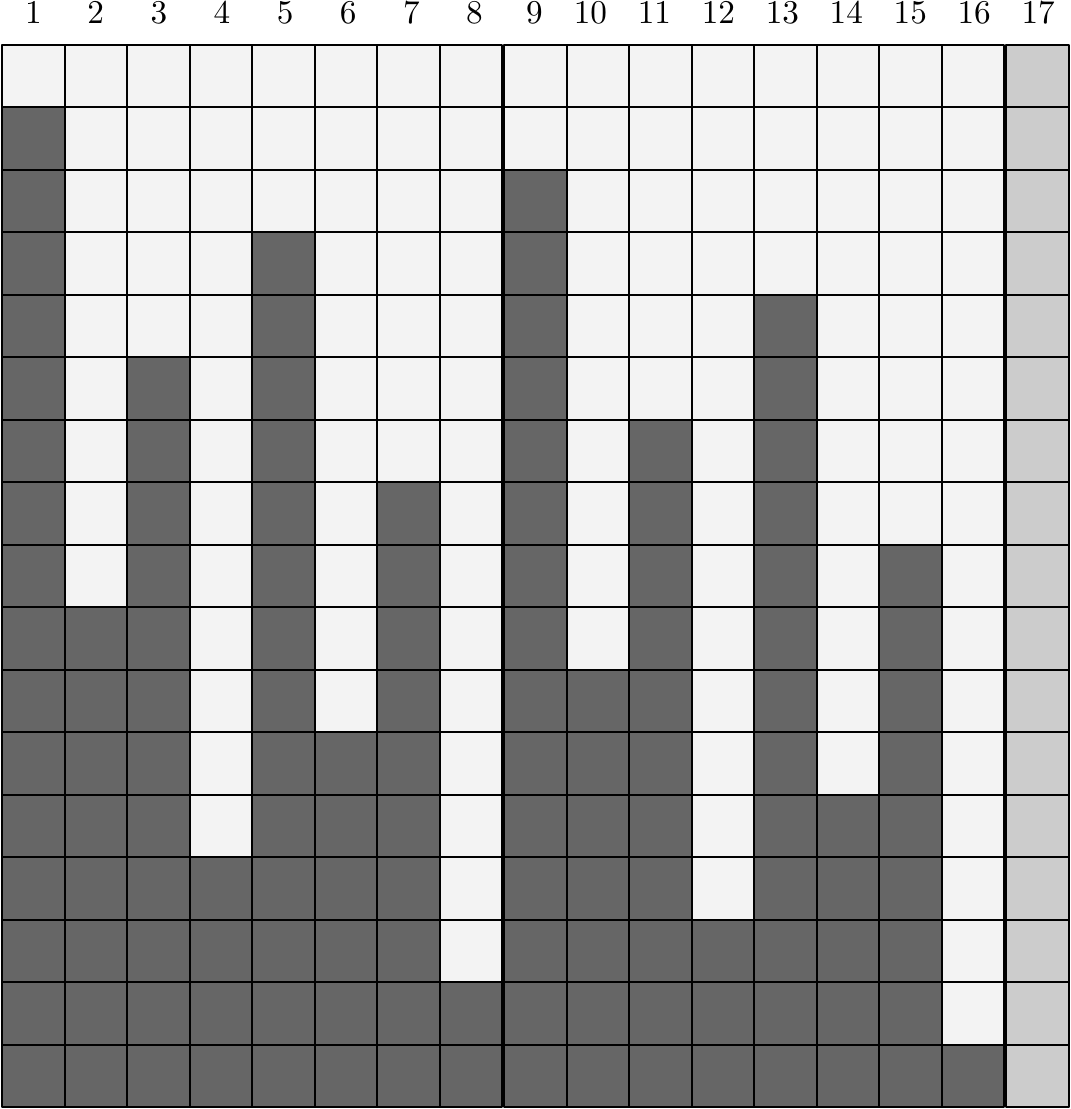}
    \caption{Illustration of the $\mathcal{S}$ with $k=16$  constructed using the permutation defined in~\eqref{eq:permutation}. A light gray square in column $i$ means that the position contains $a_i$, a dark gray square means that the position contains $b_i$. The last position in row $j$ is $\#_j$}
    \label{fig:string_as_matrix}
\end{figure}


\section{Construction with Binary Strings}
\label{sec:binary-charging}

To encode the alphabet from the previous section we first assign fresh integer labels in order of first appearance. That is,
\[
a_i=i-1,\qquad
\#_0=k,\qquad
b_{p_t}=k+2t-1,\qquad
\#_t=k+2t.
\]
Thus integer labels follow the order of first appearance, while each
symbol \(b_i\) is named by its coordinate.

Let \(d=\lceil\log_2(3k+1)\rceil\). If the \(d\)-bit representation of
\(j\) is \(\beta_1\ldots\beta_d\), set
\[
E(j)=1111\,(0\beta_1)\cdots(0\beta_d)\,00.
\]
The length of the codeword $E(j)$ is \(c=2d+6\).
Extend $E$ to strings by concatenation. Define our new string to be \(\mathcal{S}_B=E(\mathcal{S})\)  of length \(c(k+1)^2\).

The substring \(1111\) occurs exactly at codeword starts. Thus every occurrence of a complete \(E(j)\) determines \(j\) uniquely. 

\medskip

The following argument does not assume that phrase boundaries are codeword-aligned. In row \(u\), use bit-column coordinates from \(0\) to \(ck\) for boundaries in its array part. 
Rows start at multiples of \(L=c(k+1)\). A boundary at bit-column $x$ in row $u$ is the prefix length $uL+x$.

\subsection{LZ77 Upper Bound}  There is a legal non-overlapping LZ77 parsing of $\mathcal{S}_B$ using at most $11k+c$ phrases.

\begin{lemma}\label{lem:lz77-binary-string}
    We have,
    \[z_{\mathrm{no}}(\mathcal{S}_B) \leq 11k + c\;.\]
\end{lemma}
\begin{proof}
  We can encode \(E(0)\) in at most \(c\) one-bit phrases. We will prove that the remaining $3k$ symbols can be parsed using at most 3 phrases each. A subsequent new \(E(j)\), given the already occurring \(E(j-1)\), can be parsed using at most 3 phrases: (i) the common code prefix with $E(j-1)$, (ii) the changed 1 bit, and (iii) the trailing zeros, which can be copied from $E(0)$. Since the new symbols come in sorted order we can parse all new symbols using a total of $9k + c$ phrases. 

  Each new row needs at most two additional phrases (in addition to the new symbol) exactly as before. Hence, $\mathcal{S}_B$ can be parsed into at most $11k+c$ phrases.
\end{proof}

\subsection{LZ-End Lower Bound}

We first prove the binary analogue of the charging lemma from
Section~\ref{sec:interleaving}, allowing phrase boundaries to lie inside
codewords.

\begin{lemma}\label{lem:binary-charging}
Let \(v=[a,b]\) be a node of \(\mathcal{T}_k\), with children
\([a,m]\) and \([m+1,b]\). Let \(s<t\) be consecutive times in the
update subsequence restricted to \(v\), and suppose
\[
p_s\in[a,m],\qquad p_t\in[m+1,b].
\]
Then every legal LZ-End parsing of \(\mathcal{S}_B\) contains a phrase
in the array part of a row \(u\in[s,t]\), with boundary bit-columns
\(y,x\) satisfying
\begin{equation}\label{eq:binary-charged-phrase}
c(p_s-1) < y \le cm < x \le cb.
\end{equation}
\end{lemma}
\begin{proof}
Fix an arbitrary legal LZ-End parsing of \(\mathcal{S}_B\).
The codeword \(E(b_{p_t})\) first appears in row \(t\). There must be
a phrase boundary \(\eta_t\) strictly inside this codeword, so
\begin{equation}\label{eq:novel-codeword-boundary}
c(p_t-1) < \eta_t < cp_t.
\end{equation}
Otherwise a single phrase would contain the entire novel codeword.
Both bits have already appeared in row zero, so this phrase would be a
copy, contradicting the novelty of the codeword.

Since \(p_t\in[m+1,b]\), the boundary \(\eta_t\) lies in
\((cm,cb]\). Among all boundaries with bit-column in \((cm,cb]\)
occurring in rows \(s,\ldots,t\), choose the smallest bit-column
\(x\), and then the first row \(u\) containing that column. Let
\(y\) be its predecessor boundary, expressed relative to row \(u\).
Minimality of \(x\) gives \(y\le cm\). Initially, \(y\) might lie
before the start of row \(u\).

It remains to show that \(y>c(p_s-1)\). Suppose instead that
\(y\le c(p_s-1)\). The phrase ending at \(x\) contains the entire
codeword \(E(b_{p_s})\), since \(p_s\le m\) and \(x>cm\).
This phrase is a copy, and its source therefore contains an aligned
occurrence of \(E(b_{p_s})\). That symbol occurs only at coordinate
\(p_s\) in rows \(r\ge s\). Its occurrence fixes the translation
between source and target to \((r-u)L\), so the source endpoint is a
phrase boundary at bit-column \(x\) in row \(r\). Non-overlap implies
\(r<u\), because \(x>y\). Thus \(s\le r<u\), contradicting the
choice of \(u\) as the first row in \([s,t]\) containing boundary
column \(x\). Therefore \(y>c(p_s-1)\), proving
\eqref{eq:binary-charged-phrase}. In particular, the phrase lies in
the array part of row \(u\).
\end{proof}

We now use this lemma to count distinct phrases.

\begin{lemma}\label{lem:binary-lzend-lower-bound}
Every legal LZ-End parsing of \(\mathcal{S}_B\) has at least \(kh/2\)
phrases. In particular,
\begin{equation}\label{eq:binary-lzend-lower-bound}
z_{\mathrm{end}}(\mathcal{S}_B) \ge \frac{kh}{2}.
\end{equation}
\end{lemma}
\begin{proof}
Fix an arbitrary legal LZ-End parsing of \(\mathcal{S}_B\). For each
left-to-right event \((v,s,t)\), apply Lemma~\ref{lem:binary-charging}
and charge the event to a phrase with boundary bit-columns \(y,x\)
satisfying \eqref{eq:binary-charged-phrase} in a row \(u\in[s,t]\).

Assign each positive bit-column \(z\) the coordinate key
\(\lceil z/c\rceil\). The keys assigned to \(y\) and \(x\) lie in
opposite children of \(v\), so \(v\) is their unique lowest common
ancestor. A phrase can therefore be charged only by events at one
node. For a fixed node, the row intervals of successive left-to-right
events are disjoint, so each phrase is charged at most once.

For the bit-reversal permutation, the total number of left-to-right
events is
\[
\sum_{d=0}^{h-1} 2^d\frac{k}{2^{d+1}}=\frac{kh}{2}.
\]
Since these events are charged to distinct phrases, the parsing has at
least \(kh/2\) phrases. The parsing was arbitrary, so the bound holds
also for an optimal LZ-End parsing.
\end{proof}

\subsection{Proof of the Main Result}
\label{sec:main-proof}

We combine the two bounds from Section~\ref{sec:binary-charging} to prove
the binary separation stated in Theorem~\ref{thm:snapshot-separation}.

Let $k=2^h$ and let $\mathcal{S}_B=E(\mathcal{S})$ be the binary encoding of the snapshot
string constructed above. For the theorem, let $T_k=\mathcal{S}_B$. By construction,
\[
 |T_k|=c(k+1)^2,
 \qquad c=2\lceil\log_2(3k+1)\rceil+6.
\]
From Lemma~\ref{lem:lz77-binary-string} we have,
\[
 z_{\mathrm{no}}(T_k)\le 11k+c.
\]
The charging argument assigns all $kh/2$ left-to-right events to distinct
phrases in every legal LZ-End parsing. Consequently,
\[
 z_{\mathrm{end}}(T_k)\ge \frac{kh}{2}
 =\frac{k\log_2 k}{2}.
\]
Since $c=\Theta(\log k)$ and
$n=|T_k|=c(k+1)^2$, we have $\log n=\Theta(\log k)$. Therefore,
\[
 \frac{z_{\mathrm{end}}(T_k)}{z_{\mathrm{no}}(T_k)}
 \ge \frac{k\log_2 k}{2(11k+c)}
 =\Omega(\log k)
 =\Omega(\log n).
\]
The lower bound applies to every legal LZ-End parsing, so it holds in
particular when all phrase boundaries and source occurrences are chosen
optimally. This concludes the proof of Theorem~\ref{thm:snapshot-separation}.

\section{Phrases with an Appended Character}
\label{sec:appended-character}

We extend the separation to the convention in which each phrase has the
form $Xa$, where $a$ is one character and the copied prefix $X$ may be
empty. A nonempty $X$ must have an occurrence wholly before the phrase
starts. For LZ-End, this occurrence must end at an earlier phrase
boundary. We also allow the final phrase to consist solely of a copy;
the argument applies equally when the final phrase must have an appended
character. Let $z_{\mathrm{end}}^{+}(S)$ be the minimum number of phrases
in such an LZ-End parsing, and let $z_{\mathrm{no}}^{+}(S)$ be the number
of phrases in the greedy non-overlapping LZ77 parsing under this
convention. As for whole-copy LZ77, the greedy parsing minimizes the
number of phrases.

\begin{lemma}\label{lem:appended-character-conversion}
Every copy-plus-character LZ-End parsing with $z$ phrases can be
converted into a legal whole-copy LZ-End parsing with at most $2z$
phrases. Consequently, for every string $S$,
\[
z_{\mathrm{end}}(S)\le 2z_{\mathrm{end}}^{+}(S).
\]
\end{lemma}
\begin{proof}
We will convert the copy-plus-character parsing into a whole-copy parsing as follows.
Split each phrase $Xa$ into its nonempty copied prefix $X$, if any,
and the singleton $a$. Every original phrase boundary remains a boundary
in the resulting parsing. Hence the source of each copied prefix $X$
still ends at a phrase boundary and is wholly before its target.

The first occurrence of any character cannot lie inside a copied prefix,
since its source would contain an earlier occurrence. It must therefore
be an appended character and becomes a singleton phrase after the
split. This singleton is a legal literal. Every later appended occurrence
of the same character can be represented by a one-character copy of
that first singleton, whose endpoint is an earlier phrase boundary.
Thus all resulting singleton phrases are legal as well.
If the final phrase is a pure copy, retain it unchanged.
Each original phrase contributes at most
two whole-copy phrases, proving the claim. Applying the conversion to an optimal
copy-plus-character parsing gives the stated inequality.
\end{proof}

\begin{corollary}\label{cor:appended-character-separation}
For the binary strings $T_k$ of Theorem~\ref{thm:snapshot-separation},
\[
z_{\mathrm{no}}^{+}(T_k)\le c+11k,
\qquad
z_{\mathrm{end}}^{+}(T_k)\ge \frac{k\log_2 k}{4}.
\]
Consequently, for $n=|T_k|$,
\[
\frac{z_{\mathrm{end}}^{+}(T_k)}{z_{\mathrm{no}}^{+}(T_k)}
=\Omega(\log n),
\]
even over a binary alphabet.
\end{corollary}
\begin{proof}
Lemma~\ref{lem:appended-character-conversion} and
Theorem~\ref{thm:snapshot-separation} give
\[
z_{\mathrm{end}}^{+}(T_k)
\ge \frac12 z_{\mathrm{end}}(T_k)
\ge \frac{k\log_2 k}{4}.
\]

For the LZ77 upper bound, take the whole-copy non-overlapping parsing
from Section~\ref{sec:binary-charging}. For each copied phrase, copy
all but its last character from the corresponding prefix of its source,
and append its last character. The shortened source is still wholly
before the target. A one-character phrase is represented by an empty
copy and its appended character. This preserves the phrase boundaries
and the number of phrases, and is a legal non-greedy LZ77 parsing.  
Therefore
$z_{\mathrm{no}}^{+}(T_k)\le c+11k$ by the optimality of greedy LZ77.

Since $c=\Theta(\log k)$ and $n=c(k+1)^2$, we obtain
\[
\frac{z_{\mathrm{end}}^{+}(T_k)}{z_{\mathrm{no}}^{+}(T_k)}
\ge \frac{k\log_2 k}{4(c+11k)}
=\Omega(\log k)
=\Omega(\log n).
\]
\end{proof}

\bibliographystyle{plainurl}
\bibliography{references}

\begin{thebibliography}{10}

\bibitem{Akagi2023}
Tooru Akagi, Mitsuru Funakoshi, and Shunsuke Inenaga.
\newblock Sensitivity of string compressors and repetitiveness measures.
\newblock {\em Information and Computation}, 291:104999, 2023.
\newblock \href {https://doi.org/10.1016/j.ic.2022.104999}
  {\path{doi:10.1016/j.ic.2022.104999}}.

\bibitem{Bannai2023}
Hideo Bannai, Mitsuru Funakoshi, Kazuhiro Kurita, Yuto Nakashima, Kazuhisa
  Seto, and Takeaki Uno.
\newblock Optimal {LZ-End} parsing is hard.
\newblock In {\em 34th Annual Symposium on Combinatorial Pattern Matching (CPM
  2023)}, volume 259 of {\em Leibniz International Proceedings in Informatics
  (LIPIcs)}, pages 3:1--3:11. Schloss Dagstuhl -- Leibniz-Zentrum f{\"u}r
  Informatik, 2023.
\newblock \href {https://doi.org/10.4230/LIPIcs.CPM.2023.3}
  {\path{doi:10.4230/LIPIcs.CPM.2023.3}}.

\bibitem{Bille2018}
Philip Bille, Travis Gagie, Inge~Li G{\o}rtz, and Nicola Prezza.
\newblock A separation between {RLSLPs} and {LZ77}.
\newblock {\em Journal of Discrete Algorithms}, 50:36--39, 2018.
\newblock \href {https://doi.org/10.1016/j.jda.2018.09.002}
  {\path{doi:10.1016/j.jda.2018.09.002}}.

\bibitem{BonehGawrychowski2026}
Itai Boneh and Pawe{\l} Gawrychowski.
\newblock Random access to {LZ-End}: Faster and deterministic, 2026.
\newblock URL: \url{https://arxiv.org/abs/2607.14923}, \href
  {https://arxiv.org/abs/2607.14923} {\path{arXiv:2607.14923}}, \href
  {https://doi.org/10.48550/arXiv.2607.14923}
  {\path{doi:10.48550/arXiv.2607.14923}}.

\bibitem{Charikar2005}
Moses Charikar, Eric Lehman, Ding Liu, Rina Panigrahy, Manoj Prabhakaran, Amit
  Sahai, and abhi shelat.
\newblock The smallest grammar problem.
\newblock {\em IEEE Transactions on Information Theory}, 51(7):2554--2576,
  2005.
\newblock \href {https://doi.org/10.1109/TIT.2005.850116}
  {\path{doi:10.1109/TIT.2005.850116}}.

\bibitem{FS2026}
Yuto Fujie and Hiroki Shibata.
\newblock Sensitivity and size relationships of the lempel–ziv factorization,
  2026.
\newblock URL: \url{https://doi.org/10.48550/arXiv.2608.03351}, \href
  {https://arxiv.org/abs/2608.03351} {\path{arXiv:2608.03351}}, \href
  {https://doi.org/0.48550/arXiv.2608.03351}
  {\path{doi:0.48550/arXiv.2608.03351}}.

\bibitem{Fujie2025}
Yuto Fujie, Hiroki Shibata, Yuto Nakashima, and Shunsuke Inenaga.
\newblock Tight additive sensitivity on {LZ}-style compressors and string
  attractors, 2025.
\newblock URL: \url{https://arxiv.org/abs/2506.22778}, \href
  {https://arxiv.org/abs/2506.22778} {\path{arXiv:2506.22778}}, \href
  {https://doi.org/10.48550/arXiv.2506.22778}
  {\path{doi:10.48550/arXiv.2506.22778}}.

\bibitem{GawrychowskiKoscheManea2023}
Pawe{\l} Gawrychowski, Maria Kosche, and Florin Manea.
\newblock On the number of factors in the {LZ-End} factorization.
\newblock In {\em String Processing and Information Retrieval (SPIRE 2023)},
  volume 14240 of {\em Lecture Notes in Computer Science}, pages 253--259.
  Springer, 2023.
\newblock \href {https://doi.org/10.1007/978-3-031-43980-3_20}
  {\path{doi:10.1007/978-3-031-43980-3_20}}.

\bibitem{HF1998}
Haripriyan Hampapuram and Michael~L Fredman.
\newblock Optimal biweighted binary trees and the complexity of maintaining
  partial sums.
\newblock {\em SIAM J. Comput.}, 28(1):1--9, 1998.

\bibitem{HuckeLohreyReh2016}
Danny Hucke, Markus Lohrey, and Carl~Philipp Reh.
\newblock The smallest grammar problem revisited.
\newblock In {\em String Processing and Information Retrieval (SPIRE 2016)},
  volume 9954 of {\em Lecture Notes in Computer Science}, pages 35--49.
  Springer, 2016.
\newblock \href {https://doi.org/10.1007/978-3-319-46049-9_4}
  {\path{doi:10.1007/978-3-319-46049-9_4}}.

\bibitem{Ideue2021}
Takumi Ideue, Takuya Mieno, Mitsuru Funakoshi, Yuto Nakashima, Shunsuke
  Inenaga, and Masayuki Takeda.
\newblock On the approximation ratio of {LZ-End} to {LZ77}.
\newblock In {\em String Processing and Information Retrieval (SPIRE 2021)},
  volume 12944 of {\em Lecture Notes in Computer Science}, pages 114--126.
  Springer, 2021.
\newblock \href {https://doi.org/10.1007/978-3-030-86692-1_10}
  {\path{doi:10.1007/978-3-030-86692-1_10}}.

\bibitem{Jez2016}
Artur Je{\.z}.
\newblock A really simple approximation of smallest grammar.
\newblock {\em Theoretical Computer Science}, 616:141--150, 2016.
\newblock \href {https://doi.org/10.1016/j.tcs.2015.12.032}
  {\path{doi:10.1016/j.tcs.2015.12.032}}.

\bibitem{KempaSaha2022}
Dominik Kempa and Barna Saha.
\newblock An upper bound and linear-space queries on the {LZ-End} parsing.
\newblock In {\em Proceedings of the 2022 Annual ACM-SIAM Symposium on Discrete
  Algorithms (SODA)}, pages 2847--2866. SIAM, 2022.
\newblock \href {https://doi.org/10.1137/1.9781611977073.111}
  {\path{doi:10.1137/1.9781611977073.111}}.

\bibitem{KosolobovValenzuela2021}
Dmitry Kosolobov and Daniel Valenzuela.
\newblock {Lempel--Ziv} parsing for sequences of blocks.
\newblock {\em Algorithms}, 14(12):359, 2021.
\newblock \href {https://doi.org/10.3390/a14120359}
  {\path{doi:10.3390/a14120359}}.

\bibitem{KN10}
Sebasti{\'a}n Kreft and Gonzalo Navarro.
\newblock {LZ77}-like compression with fast random access.
\newblock In {\em 2010 Data Compression Conference (DCC)}, pages 239--248.
  IEEE, 2010.
\newblock \href {https://doi.org/10.1109/DCC.2010.29}
  {\path{doi:10.1109/DCC.2010.29}}.

\bibitem{KN13}
Sebasti{\'a}n Kreft and Gonzalo Navarro.
\newblock On compressing and indexing repetitive sequences.
\newblock {\em Theoretical Computer Science}, 483:115--133, 2013.
\newblock \href {https://doi.org/10.1016/j.tcs.2012.02.006}
  {\path{doi:10.1016/j.tcs.2012.02.006}}.

\bibitem{PD2006}
Mihai P\v{a}tra\c{s}cu and Erik~D. Demaine.
\newblock Logarithmic lower bounds in the cell-probe model.
\newblock {\em SIAM J. Comput.}, 35(4):932--963, 2006.
\newblock Announced at SODA 2004.

\bibitem{Rytter2003}
Wojciech Rytter.
\newblock Application of {Lempel--Ziv} factorization to the approximation of
  grammar-based compression.
\newblock {\em Theoretical Computer Science}, 302(1--3):211--222, 2003.
\newblock \href {https://doi.org/10.1016/S0304-3975(02)00777-6}
  {\path{doi:10.1016/S0304-3975(02)00777-6}}.

\bibitem{Wilber1989}
Robert Wilber.
\newblock Lower bounds for accessing binary search trees with rotations.
\newblock {\em SIAM J. Comput.}, 18(1):56--67, 1989.

\bibitem{lz77}
Jacob Ziv and Abraham Lempel.
\newblock A universal algorithm for sequential data compression.
\newblock {\em IEEE Transactions on Information Theory}, 23(3):337--343, 1977.
\newblock \href {https://doi.org/10.1109/TIT.1977.1055714}
  {\path{doi:10.1109/TIT.1977.1055714}}.

\end{thebibliography}

\end{document}